\documentclass[11pt,a4paper]{article}

\usepackage[hidelinks]{hyperref}
\usepackage{amssymb, amsmath, graphicx, subcaption}
\usepackage{multirow}
\usepackage{float}
\usepackage{fancyhdr}
\usepackage{amsthm}

\newtheorem{theorem}{Theorem}
\newtheorem{lemma}{Lemma}
\newtheorem{corollary}{Corollary}

\fancypagestyle{firststyle}
{
	\fancyhf{}
	\fancyfoot[LE,RO]{Page \thepage}
}

\title{On the Theory of Uniform Satellite Constellation Reconfiguration}
	
\author{David Arnas\thanks{Purdue University, IN, USA. Email: \textsc{darnas@purdue.edu}}, Richard Linares\thanks{Massachusetts Institute of Technology, MA, USA. Email: \textsc{linaresr@mit.edu}}}

\begin{document}

\date{}	

\maketitle

\thispagestyle{firststyle}

\begin{abstract}
This work focuses on the study of the reconfiguration strategies available for uniformly distributed satellite constellations and slotting architectures. Particularly, this manuscript deals with the cases of reducing, maintaining, and also increasing the number of available positions for satellites in the space structure, and takes into account the potential minimum distances between spacecraft in the configuration to assure the safety of the system. To that end, several approaches to solve the reconfiguration problem are presented based on the properties of Flower Constellations, and more particularly, on the properties of uniformity and symmetries present in these uniform distributions.
\end{abstract}

\section{Introduction}

In the last decade the space sector has witnessed a dramatic increase in the number of satellites orbiting the Earth, being a tendency that is expected to continue and increase in the future. This is worsen by the fact that the vast majority of these satellites are located in LEO, making this the region with the current largest number of conjunctions between satellites. Additionally, the paradigm in satellite constellation design has completely changed in that last years. Particularly, satellite constellations have passed from space structures containing 20 to 40 satellites for the largest constellations, to megaconstellations containing hundreds or even thousands of satellites. This creates new challenges not only in satellite constellation design, but also in the tools at our disposal to study this complex systems.

Of particular interest is the study of the possibilities of reconfiguration that these space structures can provide once in orbit. This allows, for instance, to reconfigure the constellation when satellites from it have failed, or to integrate new satellites into an existing constellation. Another application is slotting architectures. Reference~\cite{stm} showed that slotting architectures can be used effectively to reduce the space traffic management problem and assure a minimum distance between all the satellites compliant with the slotting architecture. However, as spacecraft technology advances, satellite control is also expected to improve, and thus, slots can reduce their size as time passes. This opens the question of the possibilities of performing an efficient reconfiguration of these extremely large space structures while maintaining the safety conditions for all the satellites involved in the structure.

To that end, this work focuses on the definition and study of the potential reconfiguration possibilities that satellite constellations can provide. In particular, this manuscript proposes a set of potential reconfiguration possibilities for uniformly distributed constellations, which have been shown to be the ones providing better properties for mitigating conjunctions between satellites in the structure. Examples of this kind of constellations include, for instance, Walker Constellations~\cite{WalkerSC_1} or Draim~\cite{Draim4} Constellations, which have been used extensively both in the past and currently for different missions, including the design of megaconstellations. But there are other satellite constellation designs that follow this principle such as Dufour constellations~\cite{dufour}, Rosette constellations~\cite{Ballard}, the Kinematically regular satellite networks~\cite{Mozhaev}, or even uniform distributions defined in the relative to Earth frame of reference~\cite{Time}. For this reason, this work makes use of the Flower Constellations, since they are the generalization of all possible uniform distributions of satellites that can be generated.

Flower Constellations~\cite{MortariFC,MortariSET1,MortariSET2,FCinsights} are a set of satellite design methodologies based on the properties of Number Theory that allow to analytically define the distribution of both uniform and non-uniform satellites constellations using a minimum number of distribution parameters. This allows to have a very compact formulation to define these constellations, while allowing to study the structure as a complete system. In particular, Lattice Flower Constellations~\cite{Avendano2DL,3dlfc,4dlfc} are of special interest since they deal with the problem of generating all possible uniform distributions that can be defined using a given set of satellites~\cite{ndnfc}. For these reasons, 2D Lattice Flower Constellations is used as a framework from which to build the reconfiguration techniques presented in this work. In addition, this manuscript makes use of Necklace Flower Constellations~\cite{ndnfc,2dnfc,3dnfc}, a generalization of Lattice Flower Constellations that allows to account for both uniform and non-uniform distributions. Particularly, this work will use the 2D Necklace Flower Constellation formulation~\cite{2dnfc} to provide a framework to naturally define available positions in a distribution of satellites, which is very useful when defining slotting architectures and satellite reconfigurations.

The study of the reconfiguration in satellite constellations is not new in the literature. For instance, Ref.~\cite{weckreconfig} and Ref.~\cite{davisreconfig} dealt with the problem of finding the optimal set of reconfiguration maneuvers to reconfigure a satellite constellation from a given original distribution to a defined final distribution while minimizing the fuel required. In a similar way, Ref.~\cite{lowreconfig} studied the same problem using low thrust engines. On the other hand, Ref.~\cite{FCcorrection} studied the maneuvering sequence of satellites in 2D Lattice Flower Constellations to perform a complete constellation reconfiguration while maintaining the characteristic uniformity of these structures. However, the problem considered in this work is different. Instead of finding the best combination of maneuvers to achieve a given final distribution, this manuscript deals with the problem of generating all the possible final satellite distributions that are compatible with the original constellation under different considerations while taking into account the minimum distances between satellites in the original and final architectures.

This work is presented as follows. First, a summary of the satellite constellation design methodologies used in this work is included for reference purposes. Second, the different reconfiguration techniques are introduced and studied. Particularly, four different problems are presented. In the first one, the positions of the original constellation are preserved during the reconfiguration. To that end, both a uniform and a non-uniform expansion approaches are provided, including the relation of 2D Necklace Flower Constellations with uniform expansions. In the second case of study, we allow the original satellites and slots to be repositioned in their orbital planes, while the new objects can be located in these orbits or others in the space. In the third situation studied, a complete reconfiguration is considered taking into account that satellites will require a low fuel budget to perform their reconfiguration. And finally, in the last case of study, we present the case of direct substitution of slots and direct addition of slots into an original and unaltered slotting distribution. For all this approaches, examples of application are presented to provide a deeper insight into the problem and the possibilities of these methodologies.


\section{Preliminaries}

In this section, a summary of analytical satellite constellation design formulations that are used in this manuscript is provided as a reference for the reader. In particular, we include the formulations of 2D Lattice and Necklace Flower Constellations as they serve as the base from which the theory presented in this work is built on, and Walker constellations due to their historical importance both in the literature and in engineering applications.

\subsection{2D Lattice Flower Constellations}

2D Lattice Flower Constellations~\cite{Avendano2DL} is a satellite constellation design formulation that allows to distribute satellites uniformly in the right ascension of the ascending node ($\Omega$), and the mean anomaly ($M$), while keeping the semi-major axis ($a$), eccentricity ($e$), inclination ($inc$), and argument of perigee ($\omega$) common for all the satellites of the constellation. In 2D Lattice Flower Constellations the relative distribution of the constellation is governed by three parameters, the number of orbits $L_{\Omega}$, the number of satellites per orbit $L_{M}$, and the combination number $L_{M\Omega}$ using the following expressions:
\begin{eqnarray}\label{eq:2dlfc}
\Delta\Omega_{ij} & = & \displaystyle\frac{2\pi}{L_{\Omega}}i, \nonumber \\ 
\Delta M_{ij} & = & \displaystyle\frac{2\pi}{L_{M}}\Big(j - \frac{L_{M\Omega}}{L_{\Omega}}i\Big),
\end{eqnarray}
where $\Delta \Omega_{ij}$ and $\Delta M_{ij}$ defines the relative distribution of the constellation with respect to a satellite of the configuration, and $i\in\{1,\dots,L_{\Omega}\}$ and $j\in\{1,\dots,L_{M}\}$ names each of the satellites of the constellation in their orbital planes and their position in the orbit respectively. Note also that in this formulation the combination number ($L_{M\Omega}$) controls the relative phasing between two consecutive orbital planes from the configuration.

\subsection{2D Necklace Flower Constellations}

2D Necklace Flower Constellations is a satellite constellation design formulation that answers the problem of how to distribute a group of satellites over a set of predefined available positions. This can represent, for instance, the problem of distributing satellites in a constellation following a set of mission requirements, the progressive development of a large constellation through different launches while maintaining some functionality, or the uniform distribution of satellites or patterns of spacecraft within a previously defined slotting architecture. 

To that end, 2D Necklace Flower Constellations make use of 2D Lattice Flower Constellations to define a set of uniform available position, from which a subset where the satellites will be located is selected. Particularly, Let $\mathcal{G}_{\Omega}$ and $\mathcal{G}_{M}$ be the necklaces defined in the right ascension of the ascending node, and the mean anomaly respectively. These necklaces represent which subset of orbital planes ($\mathcal{G}_{\Omega}$), and which subset of orbit positions ($\mathcal{G}_{M}$) are selected from the set of the available ones, that is $\mathcal{G}_{\Omega}(i^*)\in\{1,\dots,L_{\Omega}\}$, and $\mathcal{G}_{M}(j^*)\in\{1,\dots,L_{M}\}$. In other words, the position of a satellite in position $j^*$ and orbital plane $i^*$ of the real constellation can be related to the ones from the 2D Lattice Flower Constellation through:
\begin{eqnarray} \label{eq:necklace}
i & = & \mathcal{G}_{\Omega}(i^*) \mod(L_{\Omega}), \nonumber \\
j & = & \mathcal{G}_{M}(j^*) + S_{M\Omega}\mathcal{G}_{\Omega} \mod(L_M),
\end{eqnarray} 
where $S_{M\Omega}$ is a shifting parameter that may depend on each particular satellite of the constellation. The shifting parameter allows to control the phasing between two consecutive necklaces in a similar way the combination number $L_{M\Omega}$ was doing in Eq.~\eqref{eq:2dlfc}. It is important to note that the necklaces ($\mathcal{G}_{\Omega}$ and $\mathcal{G}_{M}$) are also subjected to the modular arithmetic as seen in the previous expression. Equation~\eqref{eq:necklace} can be introduced in Eq.~\eqref{eq:2dlfc} to obtain:
\begin{eqnarray}\label{eq:2dnfc}
\Delta\Omega_{i^*j^*} & = & \displaystyle\frac{2\pi}{L_{\Omega}}\mathcal{G}_{\Omega}(i^*), \nonumber \\ 
\Delta M_{i^*j^*} & = & \displaystyle\frac{2\pi}{L_{M}}\Big(\mathcal{G}_M(j^*) + S_{M\Omega}\mathcal{G}_{\Omega}(i^*)-\frac{L_{M\Omega}}{L_{\Omega}}\mathcal{G}_{\Omega}(i^*)\Big). 
\end{eqnarray}
which defines the distribution of the constellation. Additionally, if we are interested in obtaining congruent distributions, that is, configurations that maximize the number of symmetries, the shifting parameter has to fulfill the following condition:
\begin{equation} \label{eq:shiftingcondition}
Sym(\mathcal{G}_M) \mid S_{M\Omega}L_{\Omega} - L_{M\Omega},
\end{equation}
which reads $Sym(\mathcal{G}_M)$ divides $S_{M\Omega}L_{\Omega} - L_{M\Omega}$, where $Sym(\mathcal{G}_M)$ is the symmetry of the necklace in mean anomaly defined by:
\begin{equation}
\mathcal{G}_M \equiv \mathcal{G}_M + Sym(\mathcal{G}_M),
\end{equation} 
where $\equiv$ represents equivalent configurations~\cite{ndnfc}. This means that by definition, a symmetry of a necklace $Sym$ has to be a divisor of the number of available positions in the distribution, or in other words, $Sym(\mathcal{G}_M)|L_M$ and $Sym(\mathcal{G}_{\Omega})|L_{\Omega}$. Finally, it is important to note that in order to avoid duplicates in the configuration $S_{M\Omega}\in\{0,\dots,Sym(\mathcal{G}_M)-1\}$.

\subsection{Walker Constellations}

Walker-Delta Constellations~\cite{WalkerSC_1} is the most well-known satellite constellation design formulation and has been extensively studied in the literature and also applied in several space missions. The distribution defined by Walker constellations is based on the idea of performing an evenly distribution of satellites in circular orbits. To achieve that goal, satellites in the constellation share the same semi-major axis, eccentricity, inclination and argument of perigee, while the distribution is performed in the right ascension of the ascending node and the mean anomaly of the orbits. Particularly, a Walker constellation is defined by the following notation: $inc: \ t/y/f$, where $inc$ is the inclination of the constellation, $t$ is the number of satellites, $y$ is the number of different orbital planes, and $f \in \{0, \ldots, y-1\}$ is a parameter that controls the relative phasing in mean anomaly between two consecutive orbital planes from the constellation. That way, the relative distribution of the constellation is defined as:
\begin{eqnarray}\label{eq:walker0}
	\Delta\Omega_{ij} & = & 2\pi\displaystyle\frac{i}{y}, \nonumber \\ 
	\Delta M_{ij} & = & 2\pi\displaystyle\frac{y}{t}j + 2\pi\displaystyle\frac{f}{t}i,
\end{eqnarray}
where as before, the distribution is defined with respect to one of the satellites of the constellation. It is important to note that this distribution is essentially a particular case of a 2D Lattice Flower Constellation for circular orbits where $t = L_{\Omega}L_{M}$, $y = L_{\Omega}$, and $f = -L_{M\Omega} \mod(L_{\Omega})$. For this reason, we focus on this work on the 2D Lattice Flower Constellation formulation.


\section{Maintaining the original positions}

In this section we deal with the problem of finding all possible uniform distributions that from a former 2D Lattice Flower Constellation, can generate a constellation with a larger number of satellites in such a way that a subset of the final constellation satellites correspond to spacecraft from the original configuration. In that regard, it is important to note that all the results presented in this section can be apply to both circular and eccentric orbits since the 2D Lattice and Necklace Flower Constellations formulation is a more general result. 

Let $N_s$ be the original number of satellites of the constellation, and let $N_s' = nN_s$ be the number of satellites in the final configuration. The goal is to find all the final configurations that are uniform and that maintain the locations of the original constellation.

Let the original 2D Lattice Flower Constellation be defined by the following distribution:
\begin{eqnarray}\label{eq:2dlfcoriginal}
\Delta \Omega_{ij} & = & 2\pi\displaystyle\frac{i}{L_{\Omega}}; \nonumber \\
\Delta M_{ij} & = & 2\pi\left(\displaystyle\frac{j}{L_{M}} - \frac{L_{M\Omega}i}{L_{\Omega}L_{M}}\right);
\end{eqnarray}
where $L_{\Omega}$ is the number of different orbital planes in which the original constellation is distributed, $L_{M}$ is the number of satellites per orbit, $L_{M\Omega}$ is the configuration number, and $i\in\{1,\dots,L_{\Omega}\}$ and $j\in\{1,\dots,L_{M}\}$ names each satellite of the constellation. In a similar process, let the final constellation distribution be defined by:
\begin{eqnarray}\label{eq:2dlfcfinal}
\Delta \Omega_{i'j'}' & = & 2\pi\displaystyle\frac{i'}{L_{\Omega}'}; \nonumber \\
\Delta M_{i'j'}' & = & 2\pi\left(\displaystyle\frac{j'}{L_{M}'} - \frac{L_{M\Omega}'i'}{L_{\Omega}'L_{M}'}\right);
\end{eqnarray}
where $L_{\Omega}'$ and $L_{M}'$ are the number of orbits, and number of satellites per orbit of the final distribution, $L_{M\Omega}'$ is the configuration number of the new constellation, and $i'\in\{1,\dots,L_{\Omega}'\}$ and $j'\in\{1,\dots,L_{M}'\}$ names each satellite of the final constellation. Note that in order to avoid duplucates in the formulation, $L_{M\Omega}'\in\{0,\dots,L_{\Omega}'-1\}$\cite{ndnfc}.

\begin{lemma}\label{lemma:LO}
	The number of orbits in the final distribution is an integer multiple of the number of orbits in the original constellation, that is, $\exists p\ni L_{\Omega}' = p L_{\Omega}$.
\end{lemma}

\begin{proof}
	Equation~\eqref{eq:2dlfcfinal} must be able to generate as a subset the same solutions from Eq.~\eqref{eq:2dlfcoriginal} for both satellite constellations to be compatible. Therefore, by focusing on the distribution on the right ascension of the ascending node of the orbits:
	\begin{equation}
	2\pi\displaystyle\frac{i}{L_{\Omega}} = 2\pi\displaystyle\frac{i'}{L_{\Omega}'}
	\end{equation}
	where the modular arithmetic of the equation has been removed since $0\leq i < L_{\Omega}$ and $0\leq i' < L_{\Omega}'$ for any satellite of the constellation. By performing some elemental algebraic operations:
	\begin{equation}
	i' = \displaystyle\frac{L_{\Omega}'}{L_{\Omega}}i, \quad \forall i,
	\end{equation}
	and thus, it is possible to derive that $p = L_{\Omega}'/L_{\Omega}$ is an integer number since the equation must hold true for any value of $i$. Therefore, $L_{\Omega}' = p L_{\Omega}$.	
\end{proof}

\begin{lemma}\label{lemma:LM}
	The number of satellites per orbit in the final distribution is an integer multiple of the number of satellites per orbit in the original constellation, that is, $L_{M}' = (n/p) L_{M}$.
\end{lemma}

\begin{proof}
	The number of satellites of a constellation is $N_s = L_{\Omega}L_{M}$ for the case of the original distribution, and $N_s' = L_{\Omega}'L_{M}'$ for the final distribution. Additionally, $N_s' = nN_s$ where $n$ is an integer number. Therefore:
	\begin{equation}
	N_s' = L_{\Omega}'L_{M}' = p L_{\Omega} L_{M}' = n N_s = n L_{\Omega} L_{M},
	\end{equation}
	and thus,
	\begin{equation}
	L_{M}' = \left(\displaystyle\frac{n}{p}\right) L_{M}.
	\end{equation}
\end{proof}

\begin{theorem} \label{theorem:ij}
	The positions that the original satellites occupy in the final constellation relate to the original positions through the following relations:
	\begin{eqnarray}
		i' & = & \left(p\right) i; \nonumber \\
		j' & = & \left(\displaystyle\frac{n}{p}\right) j - \left(\frac{nL_{M\Omega} - pL_{M\Omega}'}{pL_{\Omega}}\right)i \mod\left(\frac{n}{p}L_{M}\right);
	\end{eqnarray}
	where $p|n$, and $L_{M\Omega}'$ must fulfill the compatibility condition:
	\begin{equation}
	\displaystyle\frac{n}{p}L_{M\Omega} - L_{M\Omega}' = BL_{\Omega};
	\end{equation}
	being $B$ an unknown integer number.
\end{theorem}

\begin{proof}
	The distribution generated by $\{\Delta\Omega_{ij}, \Delta M_{ij}\}$ must be contained as a subset of the available positions generated by $\{\Delta\Omega_{i'j'}', \Delta M_{i'j'}'\}$. Therefore, for some combinations of $\{i', j'\}$, and for every combination of $\{i, j\}$ the following conditions must be fulfilled:
	\begin{eqnarray}
	2\pi\displaystyle\frac{i}{L_{\Omega}} & = & 2\pi\displaystyle\frac{i'}{L_{\Omega}'} = 2\pi\displaystyle\frac{i'}{p L_{\Omega}}, \nonumber \\
	2\pi\left(\displaystyle\frac{j}{L_{M}} - \frac{L_{M\Omega}i}{L_{\Omega}L_{M}}\right) & = & 2\pi\left(\displaystyle\frac{j'}{L_{M}'} - \frac{L_{M\Omega}'i'}{L_{\Omega}'L_{M}'}\right) + 2\pi A \nonumber \\
	& = & 2\pi\left(\displaystyle\frac{j'}{(n/p)L_{M}} - \frac{L_{M\Omega}'i'}{n L_{\Omega}L_{M}}\right) + 2\pi A,
	\end{eqnarray}
	where $A$ is an unknown integer number generated by the result of the modular arithmetic affecting the distribution in mean anomaly. Both expressions can be transformed into two coupled Diophantine equations:
	\begin{eqnarray}
	(p) i & = & i', \nonumber \\
	nL_{\Omega}j - nL_{M\Omega}i & = & pL_{\Omega}j' - L_{M\Omega}'i' + AnL_{\Omega}L_{M},
	\end{eqnarray}
	which can be solved for $i'$ and $j'$ by substituting the value of $i'$ from the first expression into the second:
	\begin{eqnarray}
	i' & = & \left(p\right) i; \nonumber \\
	j' & = & \left(\displaystyle\frac{n}{p}\right) j - \left(\frac{nL_{M\Omega} - pL_{M\Omega}'}{pL_{\Omega}}\right)i \mod\left(\frac{n}{p}L_{M}\right).
	\end{eqnarray}
	The former expressions has to hold true for any combination of the distribution parameters $\{i,j\}$, therefore, the expressions $n/p$, and $(nL_{M\Omega} - pL_{M\Omega}')/(pL_{\Omega})$ must be both integer numbers since $j'$ is always an integer number by construction. This implies that $p|n$, that is, $p$ is a divisor of $n$. On the other hand, this condition also implies that:
	\begin{equation}
	\frac{(n/p)L_{M\Omega} - L_{M\Omega}'}{L_{\Omega}} = B;
	\end{equation}
	where $B$ is an integer number. That is, the compatibility condition between configuration numbers is:
	\begin{equation} \label{eq:Ncdiophantine}
	\displaystyle\frac{n}{p}L_{M\Omega} - L_{M\Omega}' = BL_{\Omega}.
	\end{equation}
\end{proof}

\begin{theorem}\label{theorem:Nc}
	The only possible values of $L_{M\Omega}'$ that generate different distributions and that maintain the compatibility condition are provided by:
	\begin{equation}
		L_{M\Omega}' = L_{M\Omega}'|_0 + CL_{\Omega},
	\end{equation}
	where $L_{M\Omega}'|0$ is equal to:
	\begin{equation}
	L_{M\Omega}'|_0 = \displaystyle\frac{n}{p} L_{M\Omega} \mod\left(L_{\Omega}\right),
	\end{equation}
	and $C\in\{0,\dots,p-1\}$ is the set of integer numbers that allows to define all the compatible configurations.
\end{theorem}

\begin{proof}
	From the compatibility condition between configuration numbers provided by Eq.~\eqref{eq:Ncdiophantine}:
	\begin{equation}
	L_{M\Omega}' = \displaystyle\frac{n}{p}L_{M\Omega} - B L_{\Omega},
	\end{equation}
	from which we can obtain the smallest possible value of the configuration number for the final constellation ($L_{M\Omega}'|_0$):
	\begin{equation}\label{eq:Nc0}
	L_{M\Omega}'|_0 = \displaystyle\frac{n}{p}L_{M\Omega} \mod(L_{\Omega}),
	\end{equation}
	since $B$ can represent any integer number. 
	
	The objective now is to find all possible compatible distributions. Equation~\eqref{eq:Ncdiophantine} can be rewritten as a Diophantine equation whose variables are $L_{M\Omega}'$ and $B$:
	\begin{equation}
	L_{M\Omega}' + B L_{\Omega} = \displaystyle\frac{n}{p}L_{M\Omega},
	\end{equation}
	which has always solution since $\gcd(1,L_{\Omega}) = 1$. Additionally, the solution for $L_{M\Omega}'$ of this equation is in the form:
	\begin{equation} \label{eq:Ncgeneration}
	L_{M\Omega}' = L_{M\Omega}'|_0 + C L_{\Omega},
	\end{equation}
	where $C$ is any arbitrary integer number, and $L_{M\Omega}'|_0$ is a particular solution of the equation like the one provided by Eq.~\eqref{eq:Nc0}.
	
	The configuration number of the final constellation $L_{M\Omega}'$ can have any value as long as it fulfills Eq.~\eqref{eq:Ncdiophantine}. However, we know from Ref.~\cite{ndnfc} that $L_{M\Omega}'\in\{0,\dots,L_{\Omega}'-1\}$ in order to avoid duplicates in the formulation. This means that, there are a number of compatible configuration numbers equal to $L_{\Omega}'/L_{\Omega} = p$ (possible values of $L_{M\Omega}'$ divided by step size between two consecutive compatible solutions). Therefore, duplicate configurations are avoided if and only if $C\in\{0,\dots,p-1\}$.  
\end{proof}

\begin{theorem} \label{theorem:count1}
	The number of different compatible configurations for a given expansion of the constellation by $n$ times the number of satellites is:
	\begin{equation}
	\sum_{\substack{p=1\\ p|n}}^{n} p.
	\end{equation}
\end{theorem}

\begin{proof}
	Theorem~\ref{theorem:ij} states that a solution exists if and only if $p|n$. Therefore, only the values of $p$ that are divisors of $n$ can generate compatible constellations. Additionally, and for a fixed compatible $p$, Theorem~\ref{theorem:Nc} states that $p$ different solutions always exists. Therefore, the total number of compatible configurations is the sum of all the combinations of $L_{M\Omega}'$ for each compatible value of $p$, that is:
	\begin{equation}
	\sum_{\substack{p=1\\ p|n}}^{n} p.
	\end{equation}
\end{proof}

\subsection{Relation with 2D Necklace Flower Constellations}

The expansion on number of satellites presented before can also be regarded as an application of a set of necklaces in the final distribution of the constellation. Particularly, two necklaces can be observed in the configuration, the first one applied in the right ascension of the ascending node ($\mathcal{G}_{\Omega}$), and the second one applied in the mean anomaly ($\mathcal{G}_M$). In other words, the final configuration of the constellation can be seen as the set of available positions for the constellation satellites, while the original distribution is the result of applying the necklaces into this set of available positions.

\begin{theorem}\label{theorem:necklaces}
	In the uniform expansion of a 2D Lattice Flower Constellation, and when compared to the original distribution, two necklaces are generated, one in the right ascension of the ascending node $\mathcal{G}_{\Omega}$ and the other in the mean anomaly $\mathcal{G}_{M}$. These neckaces are defined by:
	\begin{eqnarray}
	\mathcal{G}_{\Omega}(i) & = &  (p)i \mod(L_{\Omega}'), \nonumber \\
	\mathcal{G}_{M}(j) & = & \left(\displaystyle\frac{n}{p}\right) j \mod\left(L_{M}'\right).
	\end{eqnarray}
\end{theorem}

\begin{proof}
	The original constellation can be defined through the application of the 2D Necklace Flower Constellation formulation from Eq.~\eqref{eq:necklace} into the final distribution of satellites:
	\begin{eqnarray}
	i' & = & \mathcal{G}_{\Omega}(i) \mod(L_{\Omega}'), \nonumber \\
	j' & = & \mathcal{G}_{M}(j) + S\mathcal{G}_{\Omega}(i) \mod(L_{M}').
	\end{eqnarray} 
	By construction~\cite{ndnfc}, necklaces only depend on the distribution variable related to their dimension, that is, $i$ in the case of $\mathcal{G}_{\Omega}$, and $j$ in the case of $\mathcal{G}_{M}$. Therefore, these expressions can be related with the result from Theorem~\ref{theorem:ij} to obtain the equivalent distribution in the necklaces:
	\begin{eqnarray}
	\mathcal{G}_{\Omega}(i) & = &  (p)i \mod(L_{\Omega}'), \nonumber \\
	\mathcal{G}_{M}(j) & = & \left(\displaystyle\frac{n}{p}\right) j \mod\left(L_{M}'\right).
	\end{eqnarray}
\end{proof}

\begin{theorem}\label{theorem:symmetries}
	The symmetry of the necklace for $\mathcal{G}_{\Omega}$ and $\mathcal{G}_{\Omega}$ are respectively:
	\begin{eqnarray}
	\text{Sym}\left(\mathcal{G}_{\Omega}\right) & = & p, \nonumber \\
	\text{Sym}\left(\mathcal{G}_{M}\right) & = &  \displaystyle\frac{n}{p}.
	\end{eqnarray}
\end{theorem}

\begin{proof}
	 The symmetry of a necklace is defined as~\cite{ndnfc}:
	 \begin{equation}
	 \label{eq:ndsymmetry}
	 Sym(\mathcal{G}) = \min\left\{1 \leq r \leq L : \mathcal{G} + r \equiv \mathcal{G}\right\},
	 \end{equation}
	 where $\equiv$ represents equivalency between necklace distributions, and $L$ is the maximum number of available positions in the dimension where the necklace is defined. Let $i^{\mathcal{G}}$ and $j^{\mathcal{G}}$ be the available positions inside the necklaces, where $i^{\mathcal{G}}\in\{1,\dots,L_{\Omega}'\}$ and $j^{\mathcal{G}}\in\{1,\dots,L_{M}'\}$ respectively. Therefore, using Theorem~\ref{theorem:necklaces}:
	 \begin{eqnarray}
	 i^{\mathcal{G}} & = &  (p)i \mod(L_{\Omega}'), \nonumber \\
	 j^{\mathcal{G}} & = & \left(\displaystyle\frac{n}{p}\right) j \mod\left(L_{M}'\right).
	 \end{eqnarray}
	 which can be rewritten as a couple of Diophantine equations in the form:
	 \begin{eqnarray}
	 i^{\mathcal{G}} & = &  (p)i + D L_{\Omega}' = (p)i + D p L_{\Omega} = (p)(i +  D L_{\Omega}), \nonumber \\
	 j^{\mathcal{G}} & = & \left(\displaystyle\frac{n}{p}\right) j + E L_{M}' = \left(\frac{n}{p}\right) j + E\left(\frac{n}{p}\right)L_{M} \nonumber \\
	 & = & \left(\frac{n}{p}\right) \left(j + E L_{M}\right).
	 \end{eqnarray}
	 where $D$ and $E$ are two unknown integer numbers. From these expressions it is possible to derive the minimum step between two consecutive occupied positions, $\gcd(p,L_{\Omega}') = p\gcd(1,L_{\Omega}) = p$ for the case of $i^{\mathcal{G}}$, and  $\gcd(n/p,L_{M}') = (n/p)\gcd(1,L_{M}) = n/p$ for the case of $j^{\mathcal{G}}$. Since these are also the minimum step sizes produced by the movements in $i$ and $j$ respectively, this means that $p$ and $n/p$ produce the minimum rotation in order to obtain an equivalent distribution. Therefore, $\text{Sym}\left(\mathcal{G}_{\Omega}\right) = p$ and $\text{Sym}\left(\mathcal{G}_{M}\right) = n/p$.
\end{proof}

\begin{corollary}
	The final distribution from Theorem~\ref{theorem:ij} can be expressed in necklace notation with the result from Theorem~\ref{theorem:necklaces} as:
	\begin{eqnarray}
	i' & = & \mathcal{G}_{\Omega}(i) \mod\left(L_{\Omega}'\right); \nonumber \\
	j' & = & \mathcal{G}_{M}(j) - \left(\frac{(n/p)L_{M\Omega} - L_{M\Omega}'}{L_{\Omega}'}\right)\mathcal{G}_{\Omega}(i) \mod\left(L_{M}'\right).
	\end{eqnarray}
\end{corollary}

\begin{theorem}\label{theorem:2dnfc_relation}
	The equivalent shifting parameter ($S_{M\Omega}'$) of the 2D Necklace Flower Constellation that relates the original and final distributions fulfills is unique for each combination of $p$ and $L_{M\Omega}'$ and fulfills the following relation:
	\begin{equation}
	\left.\displaystyle\frac{n}{p}\right|S_{M\Omega}'L_{\Omega}' - L_{M\Omega}'.
	\end{equation}
\end{theorem}

\begin{proof}
	The condition that the shifting parameter ($S_{M\Omega}'$) has to fulfill in order for the distribution to be congruent is~\cite{ndnfc}:
	\begin{equation}
	Sym(\mathcal{G}_M) \mid S_{M\Omega}L_{\Omega} - L_{M\Omega},
	\end{equation}
	and from Theorem~\ref{theorem:symmetries} we know that the symmetry of $\mathcal{G}_M = n/p$, thus:
	\begin{equation}
	\left.\displaystyle\frac{n}{p}\right|S_{M\Omega}'L_{\Omega}' - L_{M\Omega}'.
	\end{equation}
	The former expression can be rewritten as a Diophantine equation in the form:
	\begin{equation}
	F\displaystyle\frac{n}{p} = S_{M\Omega}'L_{\Omega}' - L_{M\Omega}'.
	\end{equation}
	where $F$ is an unknown integer number. Since $\gcd(L_{\Omega}',1) = 1$, this Diophantine equation has always solution in $S_{M\Omega}'$ and $L_{M\Omega}'$, particularly:
	\begin{eqnarray}
	S_{M\Omega}' & = & S_{M\Omega}'|_0 + \lambda_1; \nonumber \\
	L_{M\Omega}' & = & L_{M\Omega}'|_0 + \lambda_1 L_{\Omega}';
	\end{eqnarray}
	where $S_{M\Omega}'|_0$ and $L_{M\Omega}'|_0$ are particular solutions of the equation and $\lambda_1$ is an integer number that allows to generate all possible combinations of the solution. Since $L_{M\Omega}'$ is defined in $L_{M\Omega}'\in\{0,\dots,L_{\Omega}-1\}$ to avoid duplicate configuration definitions, the former system of equations only provides one solution for each value of $F$ from the Diophantine equation. This means that there is only one compatible pair $\{S_{M\Omega}', L_{M\Omega}'\}$ for each value of $F$. 
	
	On the other hand, due to the definition of $S_{M\Omega}'\in\{0,\dots,(p/n)-1\}$ and $L_{M\Omega}'\in\{0,\dots,L_{\Omega}'-1\}$, the quantity $(n/p)F$ can only vary in $(n/p)F\in\{1-L_{\Omega}',\dots, (n/p-1)L_{\Omega}'\}$ in the same Diophantine equation, so, there is only $L_{\Omega}'$ possible different values that $F$ can take. In fact, by taking $F$ and $L_{M\Omega}'$ as the variables of the Diophantine equation for a given value of $S_{M\Omega}'$:
	\begin{eqnarray}
	F & = & F|_0 - \lambda_2; \nonumber \\
	L_{M\Omega}' & = & L_{M\Omega}'|_0 + \lambda_2 \displaystyle\frac{n}{p};
	\end{eqnarray}
	and since $F$ has potentially only $L_{\Omega}'$ possible values, $\lambda_2\in\{0,\dots,L_{\Omega}'-1\}$ and thus, each different value of $L_{M\Omega}'$ generates a different value of $F$ (note that the opposite is not true in general). Therefore, there is only one possible value of $F$ that relates to a pair of solutions $\{S_{M\Omega}', L_{M\Omega}'\}$, and thus, there is only one possible combination of $\{F, S_{M\Omega}'\}$ for each compatible value of $L_{M\Omega}'$. This means that $S_{M\Omega}'$ is unique for a given combination of $p$ and $L_{M\Omega}'$.
\end{proof} 

Note that the result provided by the former theorem is coherent with Theorems~\ref{theorem:Nc} and~\ref{theorem:count1} regarding of number of different compatible configurations on the expansion of the constellation. Moreover, this result allows to connect the 2D Necklace Flower Constellation formulation with the problem presented in this section. 

Another interesting problem to study is the one in which we are interested in one particular final constellation with $N_s'$ satellites where we want to find all possible original satellite distributions that allow to generate this final constellation by the process presented in this section. This means that the values of $L_{\Omega}'$, $L_{M}'$, and $L_{M\Omega}'$ are now fixed, being the parameters $n$ and $p$ the ones determining $L_{\Omega}$ and $L_{M}$ in the original distribution.

\subsection{Non-uniform expansions} \label{sec:non_uniform}

There are cases where it is more important to increase the minimum distance between satellites than providing a completely uniform distribution for the satellites in the final configuration. In general, performing an expansion on number of satellites using an completely uniform approach (as the one provided by the 2D Lattice Flower Constellations) does not provide the best results from the minimum distance between satellites perspective if the original positions are maintained. This is particularly true when the number of satellites in the final constellation or slotting architecture is very large. This effect is produced by the number of constraints that are introduced in the problem when we fix the original relative distribution of the space architecture. An example of this can be seen in the result from Theorem~\ref{theorem:count1}. Therefore, it is of interest to study how to approach the problem in case we allow including non-uniformities in the expansion of satellites.

In general, having to distribute $(N_s'-N_s)$ new satellites over the distribution requires to perform a multidimensional optimization of that amount of spacecraft, being the position of each one dependent on all the remaining satellites. This implies that the optimization can be very computationally expensive, more so when dealing with megaconstellations or slotting architectures. Fortunately, it is possible to use the symmetries from 2D Lattice Flower Constellations to obtain optimal distributions under these conditions while significantly reducing the searching space.

In a 2D Lattice Flower Constellation, there are a set of rotations of the configuration space that generate equivalent satellite distributions~\cite{ndnfc}. Particularly, if we add a first satellite to the original distribution by positioning it at coordinates $\{\Delta\Omega = \Omega_1, \Delta M = M_1\}$ and check that the minimum distance with any other satellite of the constellation is $d_1$, that means that we can position a second satellite in the coordinates:
\begin{eqnarray} \label{eq:non_uniform}
\Delta\Omega & = & \Omega_1 + \displaystyle\frac{k_1}{L_{\Omega}}; \nonumber \\
\Delta M & = & M_1 + \displaystyle\frac{2\pi}{L_{M}}\Big(k_2 - \frac{L_{M\Omega}}{L_{\Omega}}k_1\Big);
\end{eqnarray}
where $k_1$ and $k_2$ are two arbitrary integers, and still maintain exactly the same minimum distance $d_1$ with all the satellites of the constellation. This is due to the fact that each pair combination to check minimum distances between the first satellite an each of the spacecraft in the original distribution is already represented by Eq.~\eqref{eq:non_uniform}. In that sense, note that this distribution follows the same pattern as the original 2D Lattice Flower Constellation. Particularly, each pair between the first satellite and all the original satellites can be defined as:
\begin{eqnarray}\label{eq:non_uniform_1}
\Delta\Omega & = & \Omega_1 + \displaystyle\frac{i}{L_{\Omega}}; \nonumber \\
\Delta M & = & M_1 + \displaystyle\frac{2\pi}{L_{M}}\Big(j - \frac{L_{M\Omega}}{L_{\Omega}}i\Big).
\end{eqnarray}
On the other hand, the pairs defined for the second satellite are:
\begin{eqnarray}
\Delta\Omega & = & \Omega_1 + \displaystyle\frac{i}{L_{\Omega}} + \displaystyle\frac{k_1}{L_{\Omega}}; \nonumber \\
\Delta M & = & M_1 + \displaystyle\frac{2\pi}{L_{M}}\Big(j - \frac{L_{M\Omega}}{L_{\Omega}}i\Big) + \displaystyle\frac{2\pi}{L_{M}}\Big(k_2 - \frac{L_{M\Omega}}{L_{\Omega}}k_1\Big),
\end{eqnarray}
or by rearranging the equation:
\begin{eqnarray}
\Delta\Omega & = & \Omega_1 + \displaystyle\frac{i + k_1}{L_{\Omega}}; \nonumber \\
\Delta M & = & M_1 + \displaystyle\frac{2\pi}{L_{M}}\Big(j + k_2 - \frac{L_{M\Omega}}{L_{\Omega}}(i+k_1)\Big),
\end{eqnarray}
which represents exactly the same pairs from the first satellite in Eq.~\eqref{eq:non_uniform_1} but with a different ordering. Therefore, the minimum distance of the second satellite with any of the satellites from the original constellation is also $d_1$. Additionally, we also know the minimum distance between the two additional satellites. Since the additional satellite also follow the same 2D Lattice Flower Constellation distribution, this minimum distance corresponds to the overall minimum of the original constellation. 

By following this procedure, it is possible to add any number of satellites up to $N_s$, effectively doubling the satellites or slots of the final distribution. In addition, this process can be repeated as many times as required adding each time a maximum of $N_s$ satellites. This is effectively combining several 2D Lattice Flower Constellations with the same distribution parameters into the same configuration, each one with a different reference origin.

However, there is still the question of how to define these reference origins for the 2D Lattice Flower Constellations, that is, the points $\{\Delta\Omega = \Omega_1, \Delta M = M_1\}$. These can be generated by performing an optimization in the subdomain of the space. Particularly, we know from the properties of Lattice Flower Constellations~\cite{ndnfc} and from Eq.~\eqref{eq:non_uniform} that the formulation is effectively propagating a given pattern of the space defined in $\{\Delta\Omega\in[0,2\pi/L_{\Omega}), \Delta M\in[0,2\pi/L_{M})\}$ uniformly over the whole configuration. Therefore, this means that if the final distribution is comprised by $N_s'=nN_s$ satellites, the optimization will require studying only $2(n-1)$ dimensions in this subdomain.

\subsection{Examples of application}

In this subsection a series of examples are presented with the goal of showing an application of the former theorems and lemmas, and their potential uses. To that end, we consider both the cases of uniform and non-uniform architecture expansions proposed previously. These examples of application are presented in the following subsections respectively.

\subsubsection{Examples of expansion using uniform distributions}

Let an initial constellation configuration be defined by the 2D Lattice Flower Constellation parameters $L_{\Omega} = 3$, $L_{M} = 9$, and $L_{M\Omega} = 2$, a semi-major axis of $a = 29 600.137$ m and an inclination of $inc = 56^{\circ}$. This correspond to the constellation initial design parameters of the Galileo constellation. The objective is to add up to two times the number of original satellites of the constellation while maintaining the satellite positions from the original configuration. 

\begin{table}[!h]
	\centering	
	\caption{Galileo compatible 2D Lattice Flower Constellations.}
	{
		\begin{tabular}{|c|c|c|c|c|}
			\hline
			\bf{Configuration} & $\mathbf{p}$ & $\mathbf{L_{\Omega}'}$ & $\mathbf{L_{M}'}$ & $\mathbf{L_{M\Omega}'}$ \\
			\hline
			1 & 1 & 3 & 27 & 6 \\
			\hline
			2 & 3 & 9 & 9 & 2 \\
			\hline
			3 & 3 & 9 & 9 & 5 \\
			\hline
			4 & 3 & 9 & 9 & 8 \\
			\hline
		\end{tabular}
	}
	\label{tab:uniform}
\end{table}

\begin{figure}[!h]
	\centering
	{\includegraphics[width=0.9\textwidth]{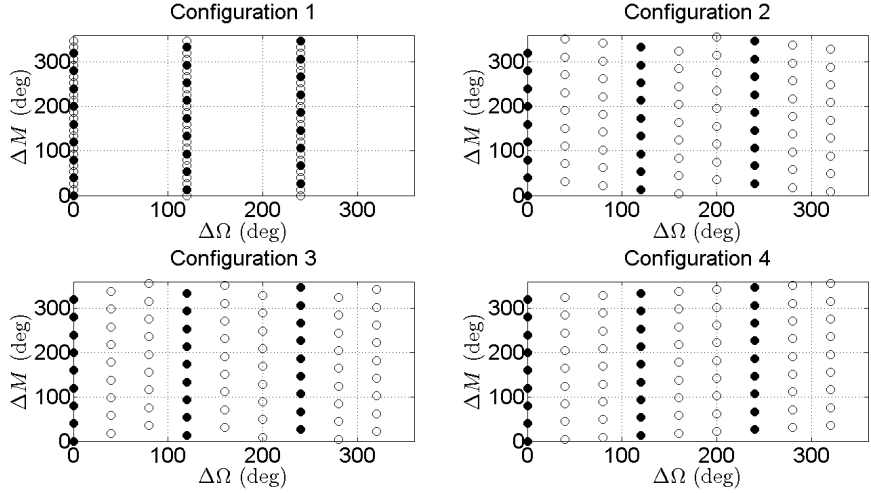}} 
	\caption{$(\Omega-M)$-space representation of the Galileo compatible 2D Lattice Flower Constellations.}
	\label{fig:uniform_compatible}
\end{figure}

Therefore, the final constellation will consist on $N_s' = 3N_s = 3L_{\Omega}L_{M} = 81$ satellites, where $n = 3$ is the number of times the constellation is expanded. This means that, based on Theorem~\ref{theorem:count1}, there are 4 compatible configurations that are completely uniform: one distributed in 3 orbital planes ($p=1$), and other 3 distributed in 9 orbital planes ($p=3$). These compatible constellations can be obtained using Theorem~\ref{theorem:Nc}, which leads to the configurations shown in Table~\ref{tab:uniform}. Moreover, Figure~\ref{fig:uniform_compatible} shows the $(\Omega-M)$-space representation of these constellations, where the circles (both white and black) are the satellites of the final configuration, while the black circles represent the satellites from the original constellation. As can be seen, the original satellites are perfectly integrated in the new structure.

\begin{table}[!h]
	\centering	
	\caption{Slotting compatible 2D Lattice Flower Constellations.}
	{
		\begin{tabular}{|c|c|c|c|c|c|}
			\hline
			\bf{Configuration} & $\mathbf{p}$ & $\mathbf{L_{\Omega}'}$ & $\mathbf{L_{M}'}$ & $\mathbf{L_{M\Omega}'}$ & $\alpha$ (deg) \\
			\hline
			1 & 1 & 246 & 14 & 202 & 0.000 \\
			\hline
			2 & 2 & 492 & 7 & 224 & 0.017 \\
			\hline
			3 & 2 & 492 & 7 & 470 & 0.304 \\
			\hline
		\end{tabular}
	}
	\label{tab:uniform_mega}
\end{table}

The advantage of using this kind of formulation is that it allows to define and study constellations with a large number of satellites using the same formalism. For instance, Ref.~\cite{stm} proposed the use of 2D Lattice Flower Constellations to define slotting architectures that maximize the minimum distance between satellites. Particularly, at $60^{\circ}$ in inclination, the 2D Lattice Flower Constellation with the largest number of spacecraft that maintains a minimum distance between satellites of at least $1^{\circ}$ is represented by $L_{\Omega} = 246$, $L_{M} = 7$, $L_{M\Omega} = 224$. The goal this time is to double the number of slots while maintaining the original positions of this slotting architecture. Table~\ref{tab:uniform_mega} shows the compatible solutions with their respective minimum distances between satellites ($\alpha$). As can be seen, one of the solutions generates a systemic conjunction (configuration 1) as a result of the distribution parameters $L_{\Omega}$, and $(L_{M} + L_{M\Omega})$ being both even numbers~\cite{fcefficient}. On the other hand, configuration 3 ($L_{\Omega}' = 492$, $L_{M}' = 7$, $L_{M\Omega}' = 470$) provides the best performance under these constraints with a minimum distance between satellites of $0.304^{\circ}$, a $30\%$ of the minimum distance of the original slotting architecture.

Figure~\ref{fig:No492_Nso7_Nc470} shows the $(\Omega-M)$-space representation of configuration 3 from Table~\ref{tab:uniform_mega}, where the black circles represent the original slots in the configuration while the white ones are the slots that have been added to the structure. This figure clearly shows that the additional slots follow the same pattern as in the original slotting architecture, not modifying the overall distribution of the structure. 

\begin{figure}[!h]
	\centering
	{\includegraphics[width=0.9\textwidth]{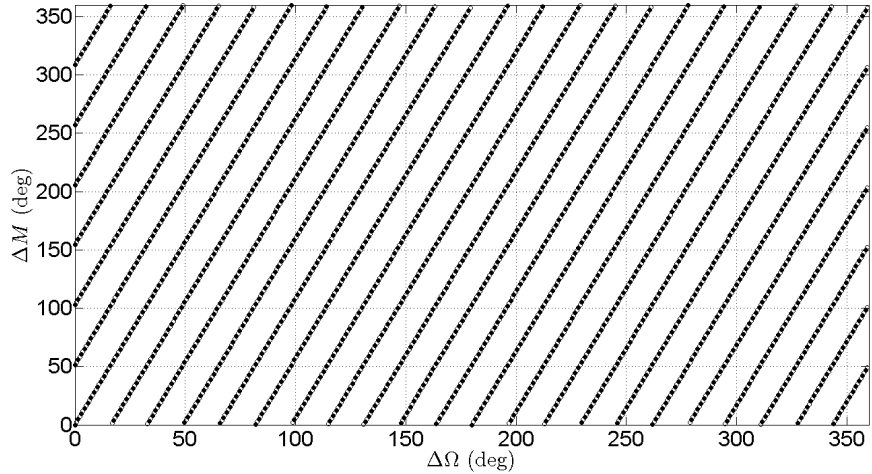}} 
	\caption{$(\Omega-M)$-space representation of the uniform compatible 2D Lattice Flower Constellation.}
	\label{fig:No492_Nso7_Nc470}
\end{figure}

This problem can also be inverted. Instead of fixing the smaller slotting architecture, we can fix the final distribution containing more slots. This situation can represent, for instance, the progressive building of an optimized large constellation in different launches. This means that the problem to solve changes to find all the possible initial configurations having half the number of slots than the final distribution that are compatible with the positions defined in the final slotting architecture. This problem is, in fact, the necklace problem~\cite{2dnfc}.

The final slotting architecture is defined as before by $L_{\Omega}' = 492$, $L_{M}' = 7$ and $L_{M\Omega}' = 470$. If the number of slots in the original distribution is half the ones in the final architecture, this means that $n = 2$, and that $p =\{1,2\}$. Then, from Theorem~\ref{theorem:symmetries} we can obtain the symmetries of the necklaces both in the right ascension of the ascending node ($\mathcal{G}_{\Omega}$) and mean anomaly ($\mathcal{G}_{M}$) depending on the value of $p$ selected. However, the symmetry of a necklace has to be a divisor of the number of available positions by definition~\cite{2dnfc}. This means that when $p = 1$, $\mathcal{G}_{M} = 2$ which is not a divisor of $L_{M}' = 7$, and thus, $p = 1$ is not a valid solution. For $p = 2$ we have $\mathcal{G}_{\Omega} = 2$ and $\mathcal{G}_{M} = 1$, and applying Theorem~\ref{theorem:2dnfc_relation} we obtain a shifting parameter equal to $S_{M\Omega}' = 0$ since $S_{M\Omega}' \in \{0,\dots,Sym(\mathcal{G}_{M}) - 1\}$~\cite{2dnfc}. Therefore we only have one possible compatible configuration where the necklace structure is defined by Theorem~\ref{theorem:necklaces}. Moreover, it is possible to define the equivalent 2D Lattice Flower Constellation. Particularly, $L_{\Omega} = L_{\Omega}'/p = 246$, $L_{M} = L_{M}'p/n = 7$, and $L_{M\Omega}$ is obtain by solving the differential equation from Theorem~\ref{theorem:Nc}:
\begin{equation}
	L_{M\Omega}' = \displaystyle\frac{n}{p}L_{M\Omega} + CL_{\Omega} \longrightarrow 470 = L_{M\Omega} + 246C,
\end{equation}
whose solution is unique and equal to $L_{M\Omega} = 224$. As can be observed, this is exactly the configuration that was used in the previous example as the original configuration. Note also that if non uniform distributions are allowed for the original distribution, there is a larger number of possible combinations as presented in Refs.~\cite{2dnfc,nominal,maintenance}.

\subsubsection{Example of expansion using non-uniform distributions}
\label{sec:exnonuniform}

In this example of application we focus on the same original slotting architecture presented in the previous example, that is, a 2D Lattice Flower Constellation at 60 deg in inclination whose distribution parameters are $L_{\Omega} = 246$, $L_{M} = 7$, $L_{M\Omega} = 224$. The goal again is to generate a compatible slotting architecture with double the number of slots, that is, the positions of the original slots are unaltered but the size of these slots will change based on the new configuration.

In order to generate this slot reconfiguration we make use of the methodology proposed in Section~\ref{sec:non_uniform} to maximize the minimum distance with these new defined slots, and the formula proposed in Ref.~\cite{fcefficient} to compute the minimum distance between two satellites located in circular orbit at the same altitude. Particularly, the optimization is performed by a brute force approach where a grid is defined in the two dimensional range $\{\Delta\Omega\in [0,2\pi/L_{\Omega}],\Delta M \in [0,2\pi/L_{M}]\}$, which corresponds to the smallest pattern in a 2D Lattice Flower Constellation as seen in Section~\ref{sec:non_uniform}. Particularly, we selected a $500\times 5000$ grid since the range in right ascension of the ascending node is smaller than the one in mean anomaly. Then, for each point in the grid, the minimum distance with all the slots of the original distribution is computed, selecting the point with largest minimum distance value as the candidate to perform the slotting expansion. Figure~\ref{fig:map_mindis} shows the map of the minimum distances as a function of the position of the slot in the pattern, where the two red dots are the positions of the original slots. For this example, the point with better minimum distance found corresponds to $\Omega_1 = 1.2995^{\circ}$ and $M_1 = 50.2251^{\circ}$, having a minimum distance of $0.5536^{\circ}$. Therefore, if this pattern defined by the original slotting architecture and this additional slot is repeated over the whole constellation, it is possible to expand the original slotting architecture to $N_s' = 2\cdot 1722 = 3444$ slots while maintaining a minimum distance of $0.5536^{\circ}$. The representation of this slotting architecture can be seen in Figure~\ref{fig:non_uniform_compatible}, where the original slots are represented by black circles, and the new ones by black circles. As can be seen, the overall distribution is not uniform if considering particular slots, but is uniform if considered the minimum size pattern characteristic from the original 2D Lattice Flower Constellation.

\begin{figure}[!h]
	\centering
	{\includegraphics[width=\textwidth]{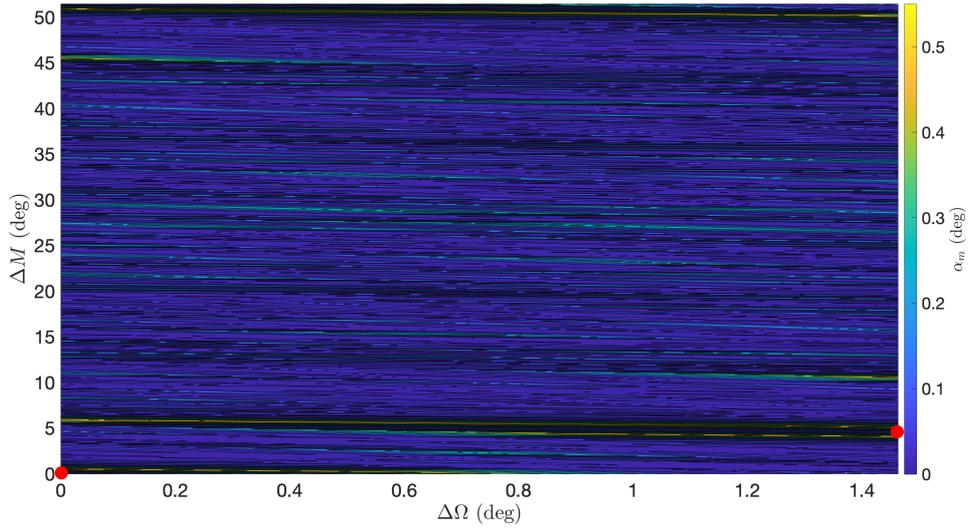}} 
	\caption{Map of minimum distances as a function of the new slot position.}
	\label{fig:map_mindis}
\end{figure}

\begin{figure}[!h]
	\centering
	{\includegraphics[width=0.9\textwidth]{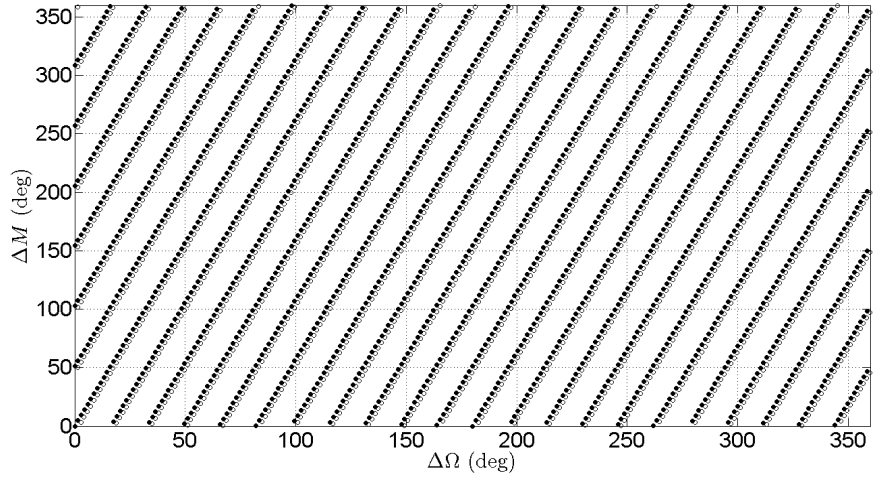}} 
	\caption{$(\Omega-M)$-space representation of the non-uniform compatible 2D Lattice Flower Constellations.}
	\label{fig:non_uniform_compatible}
\end{figure}

Compared to the uniform expansion seen before, this approach provides a larger minimum distance between satellites at the cost of losing the complete uniformity of the structure. In general, using uniform expansions implies imposing a large number of constraints into the final configuration. This can be clearly seen by the fact that the number of compatible configurations does not depend on the initial number of satellites as seen in Theorem~\ref{theorem:count1}. Nevertheless, and when considering other metrics different to minimum distance between satellites, uniform expansions can obtain much better performance. An example of this is global coverage problems as seen in Ref.~\cite{gdop}.


\section{Maintaining the orbital plane distribution}

This case of study takes into account that satellites may be able to move their positions inside their orbits, but the orbits themselves will be fixed. Note that even if the original orbits and satellites are maintained, new spacecraft can be located in other orbital planes. This can be seen as a particular case of the result of the previous section, where the condition in the mean anomaly of the satellites (and thus, the combination numbers) is relaxed.

\begin{theorem} \label{theorem:plane}
	The 2D Lattice Flower Constellations that maintain the original orbital plane distribution are defined by: $L_{\Omega}' = p L_{\Omega}$, $L_{M}' = (n/p) L_{M}$, and $L_{M\Omega}'\in\{0,1,\dots,L_{\Omega}'-1\}$, where $p|n$ is an integer number. The number of different 2D Lattice Flower Constellations under these conditions is:
	\begin{equation}
	L_{\Omega}\sum_{\substack{p=1\\ p|n}}^n p.
	\end{equation}
\end{theorem} 

\begin{proof}
	From Lemmas~\ref{lemma:LO} and~\ref{lemma:LM}, $L_{\Omega}' = p L_{\Omega}$ and $L_{M}' = (n/p) L_{M}$ where $p$ is an integer number. Moreover, from Theorem~\ref{theorem:ij} we know that $p|n$ and thus, $n/p$ is also an integer number. Additionally, since there is no other constraint in the phasing between orbital planes, this means that the combination numbers $L_{M\Omega}'$ are free to choose, that is, $L_{M\Omega}'\in\{0,1,\dots,L_{\Omega}'-1\}$. Therefore, for a given value $p$ there are $pL_{\Omega}$ different combinations. By adding the effect of all the possible distributions in number of orbital planes, the number of possible compatible constellations for $nN_s$ satellites under these condition is:
	\begin{equation}
	\sum_{\substack{p=1\\ p|n}}^n pL_{\Omega}.
	\end{equation}
\end{proof}

With these relations, we can obtain directly the equivalent distribution using the 2D Necklace Flower Constellation formulation. Particularly, Theorem~\ref{theorem:necklaces} can be used to define the necklace, Theorem~\ref{theorem:symmetries} is applied to obtain the symmetries of the necklace, and Theorem~\ref{theorem:2dnfc_relation} can be used to derive the shifting parameter that makes the distribution congruent in the space.

\subsection{Alternative approach}

Previously we have considered the case in which the number of satellites or slots in the final distribution is an integer number of times the number of satellites or slots in the original distribution. For the case of compatible expansions, this is always true do to the constrains imposed by the compatibility conditions. However, when allowing satellites to move in their orbital planes this condition is no longer required.

\begin{theorem}
    Let $N_s' = L_{\Omega}'L_{M}'$ and $N_s=L_{\Omega}L_{M}$ be the number of satellites of the final and original distributions respectively. Then, the final 2D Lattice Flower Constellations that are compatible with the original orbital plane distribution fulfill the following conditions:
    \begin{equation}
        L_{\Omega}' = p L_{\Omega}; \quad pL_M' = \displaystyle\frac{N_s'}{L_{\Omega}}; \quad L_M'\in\left[L_M,\frac{N_s'}{L_{\Omega}}\right];
    \end{equation}
    where:
    \begin{equation}
        p=\{1,2,\dots, \Big\lfloor \displaystyle\frac{N_s'}{N_s}\Big\rfloor\}. \quad \text{and} \quad L_{M\Omega}'\in\{0,\dots,L_{\Omega}'-1\}.
    \end{equation}
    These constraints provide a number of possible combinations equal to:
    \begin{equation}
        \displaystyle\frac{1}{2}\left(\Big\lfloor \displaystyle\frac{N_s'}{N_s}\Big\rfloor + 1\right)\Big\lfloor \displaystyle\frac{N_s'}{N_s}\Big\rfloor L_{\Omega}.
    \end{equation}
\end{theorem}

\begin{proof}
    Since the original orbital planes have to remain unaltered and the resultant structure has to be uniform, the number of orbital planes in the final distribution has to be a multiple of the original number of orbital planes, that is, $L_{\Omega}' = p L_{\Omega}$ (see also Lemma~\ref{lemma:LO}). Conversely, the final number of satellites per orbital plane does not have to be a multiple of the original number since we are allowing satellites to move in their orbital planes. However this introduces a constraint, $L_M' \geq L_M$ to assure that satellites and slot in the original distribution can remain in their initial orbits. Apart from that, $L_{\Omega}'$, $L_{M}'$ and $L_{M\Omega}'$ have no additional restrictions. 
    
    From the expression relating the number of satellites of the final configuration with the 2D Lattice Flower Constellation distribution parameters:
    \begin{equation}
        N_s' = L_{\Omega}'L_{M}' = pL_{\Omega}L_{M}' \longrightarrow pL_M' = \displaystyle\frac{N_s'}{L_{\Omega}}.
    \end{equation}
    Moreover, by definition, $p$ is a strict positive integer number, which in combination with the constraint $L_M' \geq L_M$ imposes the following condition:
    \begin{equation}
        L_M'\in\left[L_M,\frac{N_s'}{L_{\Omega}}\right],
    \end{equation}
    and thus, the possible values of $p$ are $p=\{1,2,\dots, \lfloor N_s'/N_s \rfloor\}$, where $\lfloor x \rfloor$ is the round down value of x.
    
    Now, the objective is to derive the number of possible combinations resultant of applying these constraints into the 2D Lattice Flower Constellation formulation. We know that $p$ can only have very specific values, particularly, $p=\{1,2,\dots, \lfloor N_s'/N_s \rfloor\}$, where each of these values defines a different distribution in orbital planes for the space structure. Additionally, since satellites are allowed to be repositioned in their orbital planes, the configuration number has no restriction and thus, $L_{M\Omega}'\{0,\dots,L_{\Omega} - 1\}$ for each value of $p$. Therefore, the number of possible combinations is:
    \begin{equation}
        \sum_{p=1}^{\lfloor N_s'/N_s \rfloor} pL_{\Omega} = L_{\Omega} \sum_{p=1}^{\lfloor N_s'/N_s \rfloor} p = \displaystyle\frac{1}{2}\left(\Big\lfloor \displaystyle\frac{N_s'}{N_s}\Big\rfloor + 1\right)\Big\lfloor \displaystyle\frac{N_s'}{N_s}\Big\rfloor L_{\Omega}.
    \end{equation}
\end{proof}

\subsection{Example of application for orbital plane maintenance} \label{sec:explane}

As in the example from previous sections, we select as the original distribution the 2D Lattice Flower Constellation defined by $L_{\Omega} = 246$, $L_{M} = 7$, $L_{M\Omega} = 224$. The goal is to find a final distribution with double the number of slots that maximizes the minimum distance between satellites while maintaining the original satellites in their orbital planes.

\begin{figure}[!h]
	\centering
	{\includegraphics[width=0.9\textwidth]{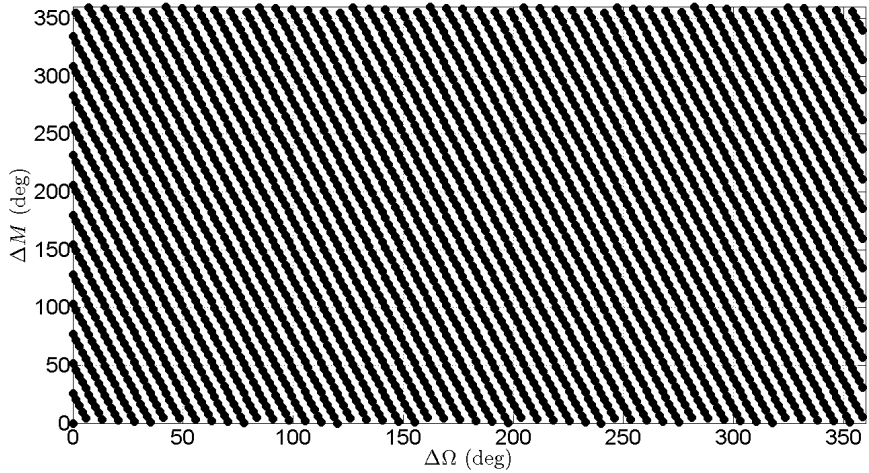}} 
	\caption{$(\Omega-M)$-space representation of the 2D Lattice Flower Constellation expansion maintaining the distribution in orbital planes.}
	\label{fig:plane}
\end{figure}

From Theorem~\ref{theorem:plane} we can derive that the number of possible different configurations compatible with the constraints is $246\cdot 3 = 738$. If all these constellations are checked in terms of minimum distance, we obtain that the configuration that presents a largest minimum distance between satellites is the one defined by $L_{\Omega} = 246$, $L_{M} = 14$, $L_{M\Omega} = 51$ with a minimum distance between spacecraft of $0.3909^{\circ}$. Figure~\ref{fig:plane} shows the $(\Omega-M)$-space representation of the this slotting architecture. Compared to the result from the example from Section~\ref{sec:exnonuniform} we can see that this result is worse in terms of minimum distance than the one obtained by doing the non-uniform distribution. This happens for two main reasons. First, the original constellation was already optimized for minimum distance for a given number of satellites. This means that increasing the number of satellites in the structure radically changes the moments in which the intersections between orbits happen, and thus, it imposes a constraint that limits the performance of this approach if we fix the original orbital planes. Second, when having a large number of satellites, the minimum distance between spacecraft is very sensitive to small changes in inclination as already pointed out in Ref.~\cite{stm}. For instance, if instead of maintaining the original inclination of $60^{\circ}$, we change it slightly to $59.2^{\circ}$ we can obtain a structure with a much better performance. Particularly, if $L_{\Omega} = 492$, $L_{M} = 7$, $L_{M\Omega} = 122$, a minimum distance of $0.5544^{\circ}$ is obtained, which slightly outperforms the result from Section~\ref{sec:exnonuniform}.

In addition, we can study the maximum angular separation that a spacecraft should move in order to get repositioned from the original configuration to the final slotting architecture. This allows to study the fuel budget required should an impulsive maneuver is performed to change both mean anomaly and right ascension of the ascending node of the satellites. If that is the case, and for a worst case scenario where all the original constellation is full, a maximum angle of $0.145^{\circ}$ is obtained for the case of $L_{\Omega} = 246$, $L_{M} = 14$, $L_{M\Omega} = 51$ and $inc = 60^{\circ}$, and $0.112^{\circ}$ for the case of $inc = 59.2^{\circ}$ and $L_{\Omega} = 492$, $L_{M} = 7$, $L_{M\Omega} = 122$. As can be seen these angular distances are small compared to the minimum distance between satellites. Additionally, and when doing reconfigurations, it is important that it is not necessary to move the satellites to the center of the new slots as long as they can maintain their dynamics inside their new slots. This allows to reduce this maximum distance even further. Note also that this maneuver happens inside the original and the final minimum distance and thus, it is possible to perform this maneuver safely within the structure. Another interesting result to note is that the maximum angle is smaller for the case of changing the inclination of the slotting architecture. This means that there are cases where a change in inclination reduces the reconfiguration fuel requirements (under the assumption of completely impulsive reconfiguration) as well as the capacity performance of the system.


\section{Complete reconfiguration}

In this expansion and reconfiguration we consider the case in which any satellite from the original distribution is able to reposition itself in the constellation including its relative position in mean anomaly and right ascension of the ascending node. This kind of reconfiguration can be performed by using the differential mean motion and plane shift of each satellite when the spacecraft has a slightly different semi-major axis than the constellation reference. Note also that this approach is equivalent to find all possible uniform distributions that can be obtained with $N_s'$ number of satellites.

\begin{figure}[!h]
	\centering
	{\includegraphics[width=0.9\textwidth]{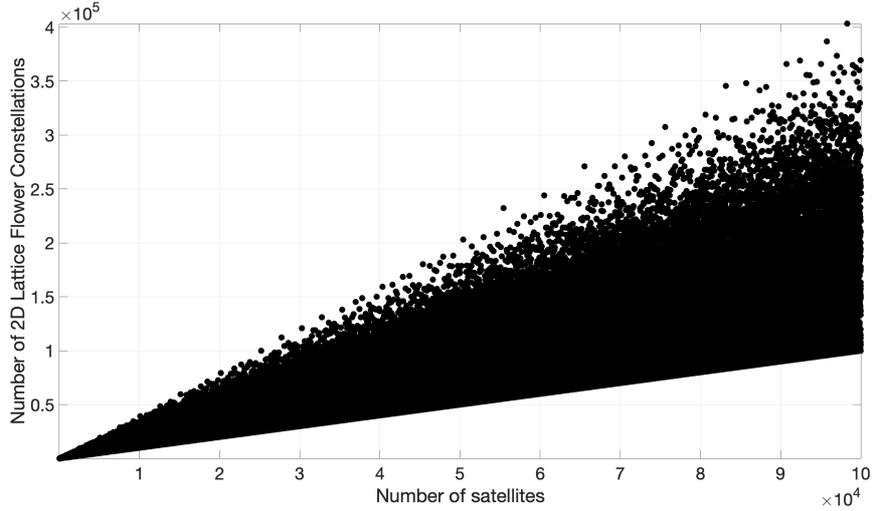}} 
	\caption{Number of different 2D Lattice Flower Constellations as a function of the number of satellites.}
	\label{fig:count_complete_reconfiguration}
\end{figure}

\begin{theorem} \label{theorem:complete_reconfiguration}
	The number of different 2D Lattice Flower Constellation distributions that can be generated with a set of $N_s'$ satellites is:
	\begin{equation}
	\sum_{\substack{L_{\Omega}'=1\\ L_{\Omega}'|N_s'}}^{N_s'} L_{\Omega}'.
	\end{equation}
\end{theorem} 

\begin{proof}
	Let $L_{\Omega}'$ be the number of final different orbital planes in which the constellation is distributed for a given configuration. Since the distribution is uniform and follow the 2D Lattice Flower Constellation formulation, $N_s' = L_{\Omega}'L_{M}'$ where $L_{\Omega}'$ and $L_{M}'$ are integer numbers. This means that $L_{\Omega}'$ and $L_{M}'$ are divisors of $N_s'$, and thus, $L_{\Omega}'|N_s'$ and $L_{M}'|N_s'$. Moreover, $L_{M\Omega}'\in\{0,1,\dots,L_{\Omega}'-1\}$ since there is no restriction in the relative phasing of satellites in their orbital planes. Therefore, for each value of $L_{\Omega}'$ there are $L_{\Omega}'$ different possibilities of configuration of the constellation. Thus, the total number of different uniform distributions for $N_s'$ satellites is the sum of the possibilities generated by the values of $L_{\Omega}'$ that are compatible, that is:
	\begin{equation}
	\sum_{\substack{L_{\Omega}'=1\\ L_{\Omega}'|N_s'}}^{N_s'} L_{\Omega}'.
	\end{equation}
\end{proof}

Figure~\ref{fig:count_complete_reconfiguration} shows the variation of the number of possible combinations as a function of the number of satellites of the constellation. From the figure it can be observed that there is a lower bound in the distribution defined by a number of combinations equal to $N_s + 1$ that happens when the number of satellites in the constellation is a prime number. Moreover it can also be noted that the number of possible combinations is in the order of $\mathcal{O}(N_s)$. Note also that this evolution can be used to assess the results from Theorems~\ref{theorem:count1} and~\ref{theorem:plane} since the summation has the same form.

\subsection{Example of application to complete reconfiguration}

As with the previous examples, the original 2D Lattice Flower Constellation considered in the one with distribution parameters $L_{\Omega} = 246$, $L_{M} = 7$, $L_{M\Omega} = 224$. However, in this case, we are going to approach the problem a bit differently, instead of setting the number of satellites of the final distribution, we impose a minimum distance that has to be fulfilled by all the satellites of the constellation. To that end, and for comparison purposes, we select the minimum distance obtained in the example of Section~\ref{sec:exnonuniform}, that is, $0.5536^{\circ}$, and perform a brute force search of all the possible 2D Lattice Flower Constellations as described in Ref.~\cite{stm}. The result of this search is a 2D Lattice Flower Constellation with the following distribution parameters: $L_{\Omega} = 4243$, $L_{M} = 1$, $L_{M\Omega} = 951$, and a minimum distance between satellites of $0.5661^{\circ}$. This final constellation contains $4243$ satellites, a $23\%$ increase when compared with the result from the example of Section~\ref{sec:exnonuniform}.

\begin{table}[!h]
	\centering	
	\caption{Best 2D Lattice Flower Constellations close to $inc = 60^{\circ}$.}
	{
		\begin{tabular}{|c|c|c|c|c|c|}
			\hline
			\bf{inc (deg)} & $\mathbf{\gamma}$ \bf{(deg)} & $\mathbf{N_s'}$ & $\mathbf{L_{\Omega}'}$ & $\mathbf{L_{M}'}$ & $\mathbf{L_{M\Omega}'}$ \\
			\hline
			59.0 & 0.5545 & 4425 & 4425 & 1 & 3225 \\
			\hline
			59.2 & 0.5648 & 4285 & 857 & 5 & 207 \\
			\hline
			59.3 & 0.5539 & 4667 & 4667 & 1 & 726 \\
			\hline
			59.4 & 0.5649 & 4366 & 4366 & 1 & 444 \\
			\hline
			59.7 & 0.5597 & 4302 & 2151 & 2 & 445 \\
			\hline
			60.0 & 0.5661 & 4243 & 4243 & 1 & 951 \\
			\hline
			60.1 & 0.5642 & 4243 & 4243 & 1 & 951 \\
			\hline
			60.2 & 0.5613 & 4488 & 408 & 11 & 102 \\
			\hline
			60.3 & 0.5561 & 4341 & 4341 & 1 & 1248 \\
			\hline
			60.5 & 0.5654 & 4444 & 2222 & 2 & 909 \\
			\hline
			60.9 & 0.5623 & 4611 & 4611 & 1 & 1855 \\
			\hline
		\end{tabular}
	}
	\label{tab:complete_recon_inc}
\end{table}

The previous result can be improved even further if some variation in the inclination of the constellation is allowed. Particularly, Table~\ref{tab:complete_recon_inc} contains a list of 2D Lattice Flower Constellations that, under the condition of minimum distance of at least $\gamma = 0.5536^{\circ}$, are able to contain more than $4243$ satellites. For this computation, only increments of $0.1^{\circ}$ in inclination where considered ranging from $inc\in[59^{\circ},60^{\circ}]$, where an exhaustive search of 2D Lattice Flower Constellations was performed for each of these inclinations. These results show that by allowing small changes in inclination it is possible to increase the number of satellites up to $10\%$ for this particular example. This shows the sensibility that large constellations and slotting architectures have under variations in the inclination, and how this sensibility can be used to improve a given architecture design.

If instead of allowing to change the number of satellites in the final distribution, we fix the number of satellites to double as in previous examples, the best 2D Lattice Flower Constellation that is obtained is $L_{\Omega} = 3440$, $L_{M} = 1$, $L_{M\Omega} = 92$, which has a minimum distance between satellites of $0.4438^{\circ}$. Its $(\Omega-M)$-space representation can be seen in Figure~\ref{fig:complete}. Compared to previous results at the same inclination, this solution performs better than the one from Section~\ref{sec:explane} (note that this solution belongs to the set studied in a complete reconfiguration), but worse than the one from Section~\ref{sec:exnonuniform}. This is due to the fact that fixing the number of satellites and the inclination still impose important constraints in the constellation design, as they control when and where the intersections between satellites happen. If, as done previously, we allow some variation in the inclination, this result can be improved. Particularly, if an inclination of $59.2^{\circ}$ is selected, the 2D Lattice Flower Constellation defined by $L_{\Omega} = 861$, $L_{M} = 4$, $L_{M\Omega} = 840$ has a minimum distance between satellites of $0.5671^{\circ}$, which is again better than the result obtained in previous sections.

\begin{figure}[!h]
	\centering
	{\includegraphics[width=0.9\textwidth]{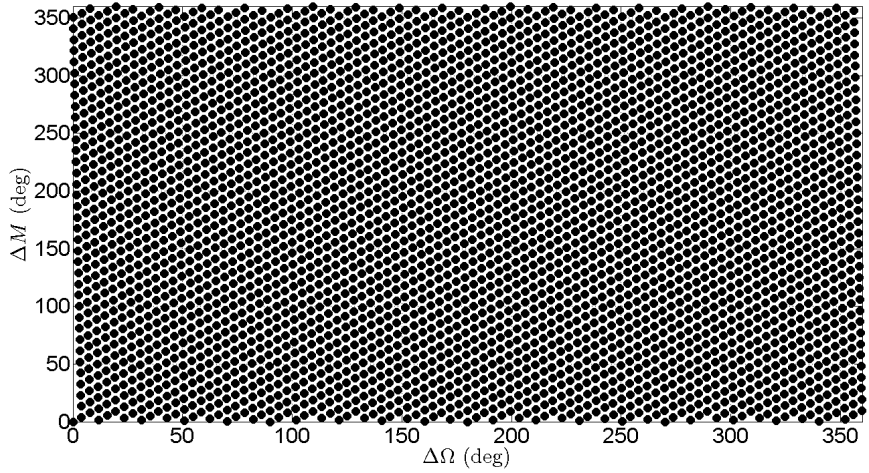}} 
	\caption{$(\Omega-M)$-space representation of the complete reconfiguration.}
	\label{fig:complete}
\end{figure}

It is important to note that, in general, performing a complete reconfiguration does not necessary guarantee the best performance when considering minimum distance between satellites if either the number of satellites, the constellation inclination, or both of them are fixed for the optimization process. This happens due to the sensitivity that minimum distances have under small changes in the inclinations and mean anomalies of satellites (which in a 2D Lattice Flower Constellation also depend on the total number of spacecraft contained in the constellation). Therefore, by allowing a variation of the inclination, even if these variations are small, and on the number of satellites of the configuration, it is possible to improve significantly the performance of the constellation as it is shown in this section.

Finally, it is also interesting to assess the maximum variation that any satellite in one of the original slots may experience under this kind of reconfiguration if it requires to maintain its coherence with the final slotting architecture. In that regard, and should we focus on changes in the right ascension of the ascending node of the orbits due to their requirements in fuel budget or time if the J2 perturbation is used to perform the maneuver, the problem then consist on finding the farthest distance between any orbital plane from the original distribution to any orbital plane from the final distribution. By doing this computation in the case of $L_{\Omega} = 3440$, $L_{M} = 1$, $L_{M\Omega} = 92$, we obtain a maximum angular distance between orbital planes of $0.0519$ deg. If the slot in that orbital plane is already occupied by another satellite in the constellation (remember that $L_{\Omega}$ can be larger than $L_{\Omega}'$ in this case of study), then the satellite would require to change its right ascension of the ascending node by $0.0528$ deg. If this also is occupied, then plane distances increase with $2\pi/L_{\Omega}'$ (the minimum inter plane distance in the final distribution). This implies that satellites will require a minimal change in their right ascension of the ascending nodes even for the worst case scenario to change to the new slotting configuration. If instead, we consider a complete impulsive maneuver where satellites change their mean anomaly and right ascension of the ascending node to reach the center of the new slot, the maximum angle that a satellite would have to move is $0.163^{\circ}$, which is again small, more so if we consider that satellites do not have to move exactly to the slot center. In that regard, note also that this distance is much smaller than the minimum distance between satellites in the final configuration.


\section{Slot substitution and addition}

The idea of this approach is to study the possibilities of expanding the initial slotting architecture when we allow slots to be smaller due to a tighter control in all the satellites compliant with the space structure. To that end, two approaches are studied, the direct slot substitution by the new ones, and the direct addition of smaller slots within the original slotting architecture. Both approaches can and should be combined to better optimize the capacity of the slotting architecture at a given altitude.

\subsection{Slot substitution}

In this approach we consider the cases in which the new defined slot size for one, several or all the slots of the architecture is several times smaller than the original slots. This allows to perform the direct substitution of the old bigger slots by the new smaller slots. Particularly, and in order to maintain the minimum distances stable through the dynamics of the satellites, the new slots have to be positioned in the same orbit than the original slot, but with a relative phasing in mean anomaly that depends on the size of the slot and the number of new slots defined. In other words, let $\rho$ be the non-dimensional size of the slot in the original distribution, that is, the slot size divided by the semi-major axis of the orbit. Then, we can use the expression provided by Ref.~\cite{mindis}:
\begin{eqnarray}\label{eq:midis}
\rho_{min} & = & 2\left|\displaystyle\sqrt{\displaystyle\frac{1+\cos^2(inc)+\sin^2(inc)\cos(\Delta\Omega)}{2}}\sin\left(\frac{\Delta F}{2}\right)\right|; \nonumber \\
\Delta F & = & \Delta M - 2\arctan\left(-\cos(inc)\tan\left(\displaystyle\frac{\Delta\Omega}{2}\right)\right),
\end{eqnarray}
to obtain the minimum distance between a pair of satellites at the same altitude and circular orbits. by using this expression, the size of the slot expressed in mean anomaly when considering the same orbit is:
\begin{equation}
\Delta M \in \left(-\arcsin\left(\displaystyle\frac{\rho}{2}\right), \arcsin\left(\displaystyle\frac{\rho}{2}\right)\right),
\end{equation}
with respect to the original center position of the slot. Therefore, it is possible to distribute $n$ smaller slots of size $\rho'$:
\begin{equation}
\rho' = 2\sin\left(\displaystyle\frac{1}{n}\arcsin\left(\frac{\rho}{2}\right)\right), 
\end{equation}
in the following relative positions:
\begin{equation}
\Delta M_k = \left(\left(\displaystyle\frac{1}{n}-1\right)+\left(k-1\right)\frac{2}{n}\right)\arcsin\left(\frac{\rho}{2}\right), 
\end{equation}
where $k\in\{1,\dots,n\}$ names each one of the new slots. Note that following this procedure the new slots are still contained in the original slot and thus, they fulfill automatically the minimum distances with the rest of the slotting architecture. Therefore, it is possible to do this direct substitution in a one to one basis, or even in the whole space structure while maintaining all the safety conditions of the system.

\subsection{Addition of slots}

In this final case of study we consider the situation where the original slotting architecture is maintained, but we include new slots that have a smaller size than the original ones from the architecture. This allows to define these new slots within the original structure without requiring to modify any of the existent slots.

To that end, we propose a methodology based on Section~\ref{sec:non_uniform}, but instead of considering just one size of slots, we take into account that each subset of slots defined by the 2D Lattice Flower Constellation formulation can have a different slot sizes. Particularly, the slots from the original distribution are left unaltered while the new slots will be smaller in size. The process consist on finding the distribution of points in $\{\Delta\Omega\in[0,2\pi/L_{\Omega}), \Delta M\in[0,2\pi/L_{M})\}$ that maximizes the minimum distance with the slots already defined in the constellation and place the new slots in these positions. This is in fact the same process defined in Section~\ref{sec:non_uniform} when the non-uniform expansion was presented. However, in this case the slot size of the new slots is different to the original ones and depends on the number of additional slots added in each pattern of the constellation. For instance, in the case of adding an additional slot per constellation pattern (which implies doubling the number of slots of the configuration), the maximum slot size of these new slots is equal to:
\begin{equation}
\rho' = 2\left(\alpha_m - \displaystyle\frac{\rho}{2}\right),
\end{equation}
where $\rho'$ and $rho$ are the angular slot sizes of the new and old slots respectively, and $\alpha_m$ is the minimum distance obtained for these points with the rest of the slotting architecture. In other cases with more slots added in each constellation pattern, this maximum slot size will depend on the relative distribution of these new slots.

\subsection{Example of substitution and addition of slots}

As in previous examples, the goal is to expand the slotting architecture defined by $L_{\Omega} = 246$, $L_{M} = 7$, $L_{M\Omega} = 224$ and double the number of slots presented in the configuration. However, in this case, the original slots will remain unaltered while the new slots will have a smaller size. From the example presented in Section~\ref{sec:exnonuniform} we know that the point in the constellation pattern that maximizes the minimum distance with the original slotting architecture is $\Omega_1 = 1.2995^{\circ}$ and $M_1 = 50.2251^{\circ}$, which presents a minimum distance of $0.5536^{\circ}$ with the original configuration. This means that it is possible to define slot sizes of $0.0942^{\circ}$ for this new slots without altering the original slots of the structure. As it can be seen, the size is very small compared to the one of the original distribution ($1.0130^{\circ}$). This is caused by the fact that the original distribution was already optimized for minimum distance and thus, slots were already filling the majority of the available space. On the other hand, the slot distribution follows the one presented in Figure~\ref{fig:non_uniform_compatible}. 

\begin{figure}[!h]
	\centering
	{\includegraphics[width=0.9\textwidth]{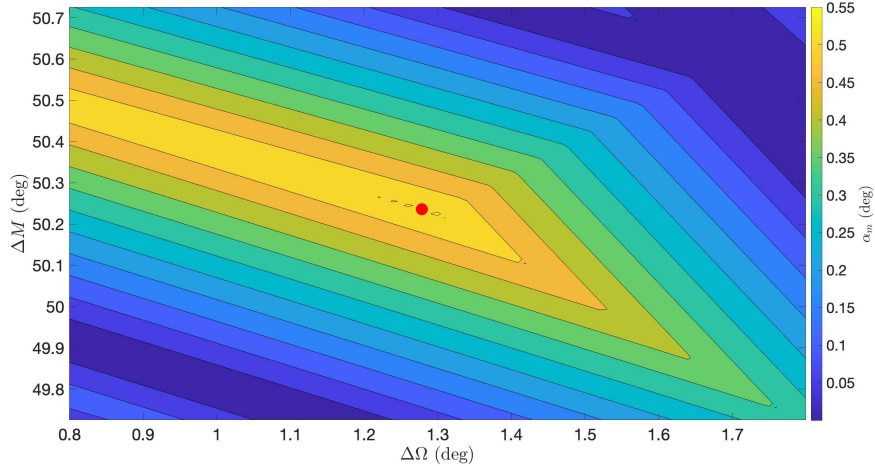}} 
	\caption{Map of the minimum distance in the region close to the optimal solution.}
	\label{fig:zoommap}
\end{figure}

Moreover, it is also interesting to study the variation of the minimum distance when modifying the position of these new slots. Figure~\ref{fig:zoommap} shows a map of the minimum distance ($\alpha_m$) from the region close this optimal point (represented by the red dot in the center of the plot). This figure provides the sensitivity of the minimum distances with small variations in either mean anomaly and right ascension of the ascending node. As can be seen, there is a region around the defined point where locating the new slots provides a good performance.


\section{Conclusions}

In this work we propose a set of methodologies for the reconfiguration of uniform satellite constellations and slotting architectures. The goal of these methodologies is to provide tools to study all the potential reconfigurations and compatible expansions that are available in these kind of space structures under certain conditions. Particularly, we deal with the problem of expanding a uniform structure while maintaining the original positions, we also study the case where the initial satellites and slots are allowed to be repositioned within their original orbits, and the case where a complete reconfiguration of the structure is performed. In addition, we also consider the case where the expansion is non-uniform with respect to the original distribution, the direct substitution of slots, and the addition of smaller slots within the original space structure. To that end, we make use of the formulations of Lattice Flower Constellations and 2D Necklace Flower Constellations to help us define and study these uniform satellite and slotting structures. 

It is important to note that the reconfiguration possibilities provided by the methodologies proposed in here can also be directly applied to elliptic orbits since at any point in the derivation of the reconfiguration theorems presented, the circular condition of the orbits was assumed. Therefore, the reconfiguration possibilities remain true for both circular and elliptic orbits, being the difference between them the different expressions that have to be used to compute the minimum distances between spacecraft and slots.

Moreover, the methodologies presented in this work can also be combined together to benefit from the properties of each of these techniques. This is especially interesting for the direct substitution and addition of slots with the rest of the proposed methodologies due to the reconfiguration possibilities that they can provide in the evolution of uniformly distributed space systems. Another important thing to note is that the defined slots do not require to only contain only one satellite. In fact, several satellites in a flight formation can be positioned inside a single or several slots as long as the spacecrafts are compliant with the rest of the slotting architecture.

The examples of reconfiguration presented in this work focus on circular orbits due to their widespread use in low Earth orbits and slotting architectures. Particularly, we selected a slotting architecture that was already optimized for minimum distance between satellites to be the space structure in from which to base these expansions. As a result of that, the final distributions presented are not optimal regarding minimum distance between satellites. This issue can be solved by performing the optimization in the final distribution or, alternatively, in both the original and final distributions simultaneously. This will improve the performance results presented in here. However, since the main goal of this work was to provide a reconfiguration framework, we opted to focus on already defined space architectures as the one used in the examples provided.

Finally, it is important to note that the techniques presented in this work are able to significantly reduce the computational time to find reconfiguration options for satellite constellations and slotting architectures. This is done both by reducing the searching space for optimization techniques, and by analytically defining all the solutions that are compatible with a set of given conditions. Due to the number of objects in slotting architectures and planned megaconstellations, this property can help study those systems more efficiently while providing deeper insight on the inherent structure contained in these space architectures.

\section{Acknowledgements}
This work is sponsored by the Defense Advanced Research Projects Agency (Grant N66001-20-1-4028), and the content of the information does not necessarily reflect the position or the policy of the Government.  No official endorsement should be inferred. 
Distribution statement A: Approved for public release; distribution is unlimited.


\end{document}